\documentclass[11pt, a4paper]{article}%

\usepackage{amssymb,amsmath,amscd,amsfonts,amsthm,bbm,mathtools}
\usepackage[colorlinks]{hyperref}
\usepackage{a4wide}
\usepackage{graphicx}
\usepackage[usenames,dvipsnames]{xcolor}
\definecolor{violet}{rgb}{0.6,0.4,0.8}
\usepackage{enumerate,multicol,longtable,array}
\usepackage{wrapfig}
\usepackage[width=.75\textwidth]{caption}

\newtheorem{theorem}{Theorem}[]

\newtheorem{proposition}{Proposition}[section]

\newtheorem{definition}{Definition}[section]
\newtheorem{remark}{Remark}[section]

\newcommand{\ind}[1]{\mathbbm{1}_{\left[ {#1} \right] }}

\newcommand{\calF}{\mathcal{F}}
\newcommand{\calG}{\mathcal{G}}

\newcommand{\calW}{\mathcal{W}}
\newcommand{\bfz}{\mathbf{z}}
\newcommand{\Z}{\mathbb{Z}}
\newcommand{\N}{\mathbb{N}}
\newcommand{\R}{\mathbb{R}}

\newcommand{\bbP}{\mathbb{P}}
\newcommand{\ep}{\varepsilon}

\newcommand{\wlim}{\lim}

\title{Examples of DLR states which are not weak limits of finite
  volume Gibbs measures with deterministic boundary conditions}
\author{Loren Coquille}
\date{\today}

\begin{document}
\maketitle

\begin{abstract} 
We review what is known about the structure of the set of weak limiting states of the
  Ising and Potts models at low enough temperature, and in particular we prove
that the mixture $\frac12(\mu^\pm+\mu^\mp)$ of two reflection-symmetric Dobrushin
  states of the 3-dimensional Ising model at low enough temperature is
  a Gibbs state which is not a limit of
  finite-volume measures with deterministic boundary conditions.
Finally we point out what the issues are in order to extend the
analysis to the Potts model, and give a few conjectures.
\end{abstract}

\section{Introduction}

In the end of the 60s, the seminal works of Dobrushin and Lanford-Ruelle~\cite{Dob68_DLR,LanRue69}
 describe the equilibrium states of a lattice model of statistical mechanics in the
thermodynamic limit as
probability measures $\mu$ that are solutions of the DLR equation:
\[
\mu(\cdot)=\int\mathrm{d}\mu(\omega)\gamma_\Lambda(\cdot\,\vert\,\omega),
\qquad\text{for all finite subsets $\Lambda$ of the lattice,}
\]
where the probability kernel $\gamma_\Lambda$ is the Gibbsian specification
associated to the system; see~\cite{Geo88}. Under very weak assumptions 
(at least for bounded spins),
it can be shown that the set $\calG$ of all DLR states is a non-empty
simplex, which contains the (a priori non-convex) set of weak limits of finite-volume Gibbs
measures, denoted $\calW$. Moreover, extremal measures of
$\calG$, the set of which is denoted $\rm{ex}\calG$, have the extra property to be weak limits of finite
volume measures with boundary conditions that are typical for it, which implies that $\rm{ex}\calG\subset\calW$. 
The analysis of $\rm{ex}\calG$ is in general a very hard problem which remains
essentially open in dimensions $3$ and higher, for any nontrivial
model, even in perturbative regimes.

This article focuses on the relationship between $\calG$ and
$\calW$. Although it is clear that $\calW\subseteq\calG$, 
it is harder to determine whether
$\calW=\calG$. For
  example, in \cite{AlbZeg1992} the question is mentioned as
an open problem. Here we will settle the question by showing that it
is not the case. Indeed we will exhibit
a (non-extremal) infinite-volume measure of the 3-dimensional Ising model
which belongs to $\calG\backslash\calW$.\\

\textit{Note added. }After submitting this paper, I was informed by
Y.\ Higuchi that the result of Theorem \ref{result-ising} was
independently found before, and privately communicated to him, by
M.~Miyamoto, who afterwards also mentioned it in his textbook (only in
Japanese) \cite{Miy2004}.\\

 We now introduce
some further notation and define the sets $\calG$ and $\calW$ in detail for the
Ising and Potts models.
Let $q,d\in\N\backslash\{0,1\}$ and $\Omega=\{1,\ldots,q\}^{\Z^d}$ be the space of configurations. Let $\Lambda$
be a finite subset of $\Z^d$, and $\Lambda^c=\Z^d\setminus\Lambda$ be its
complement. The finite-volume Gibbs measure in $\Lambda$ for the
$q$-state Potts model with boundary conditions $\omega\in\{0,1,\ldots,q\}^{\Z^d}$ and at
inverse-temperature $\beta>0$ is the probability measure on $\Omega$ (with the
associated product $\sigma$-algebra) defined by
\[
\bbP_{q,\beta,\Lambda}^\omega(\sigma)
=
\begin{cases}
\frac{1}{Z^\omega_{\beta,\Lambda}}{\rm e}^{-\beta H^\omega_\Lambda(\sigma)}& \text{ if
$\sigma_i=\omega_i$ for all $i\in \Lambda^c$}\\
0 &\text{ otherwise},
\end{cases}
\]
where the normalization constant $Z^\omega_{\beta,\Lambda}$ is the partition
function. The Hamiltonian in $\Lambda$ is given by
\[
H^\omega_\Lambda(\sigma) = -\sum_{\substack{i\sim
j\\\{i,j\}\cap\Lambda\neq\varnothing}}\delta_{\sigma_i,\sigma_j}
\]
where $i\sim j$ if $i$ and $j$ are nearest neighbors in $\Z^d$. In the case of
pure boundary condition $i \in\lbrace 1,\ldots, q\rbrace$, meaning that
$\omega_x=i$ for every $x\in \Lambda^c$, we denote the measure by $\bbP_{q,\beta,\Lambda}^{i}$.  In the case of
free boundary condition,
$\omega_x=0$ for every $x\in \Lambda^c$, we denote the measure by $\bbP_{q,\beta,\Lambda}^{\varnothing}$.

 Below we write
$\mu^\omega_{\beta,\Lambda}$ for the Ising measure on
$\{-1,+1\}^{\Z^d\cap\Lambda}$ with boundary condition $\omega$, that
is for $\bbP^\omega_{2,\beta/2,\Lambda}$ with states $1,2$ identified
with $-1,+1$ (the constant $1/2$ in front of $\beta$ comes from the
identity $\delta_{\sigma_i,\sigma_j}=(1+\sigma_i\sigma_j)/2$ when $\sigma_i,\sigma_j\in\{-1,+1\}$). 

For an arbitrary subset $A$ of $\Z^d$, let $\calF_A$ be the $\sigma$-algebra
generated by spins in $A$. 

\begin{definition}A probability measure $\bbP$ on $\Omega$
is an \emph{infinite-volume DLR state} for the $q$-state Potts model at
inverse temperature $\beta$ if and only if it satisfies the following DLR
condition:
\begin{equation}\label{def-DLR}
\bbP(\cdot|\mathcal F_{\Lambda^c})(\omega)
=
\bbP^\omega_{q,\beta,\Lambda}\qquad\text{ for $\bbP$-a.e. $\omega$,
and all finite subsets $\Lambda$ of $\Z^d$}.
\end{equation}
Let $\calG_{q,\beta}$ be the space of infinite-volume DLR
states for the $q$-state Potts model. This set being a simplex
\cite{Geo88}, let $\rm{ex}\calG_{q,\beta}$ denote the set of its
extremal points. Let $\rm{tr}\calG_{q,\beta}$ denote the set of translation
invariant DLR states, namely measures $\bbP\in\calG_{q,\beta}$ such that
$\bbP(f\circ\tau)=\bbP$ for all local functions $f$ and all
translations $\tau$ of $\Z^d$.
\end{definition}

We also formally define the (in principle smaller) set of Gibbs states
which can be obtained via boundary conditions as follows:

\begin{definition}A probability measure $\bbP$ on $\Omega$
is a \emph{weak-limiting Gibbs state} for the $q$-state Potts model at
inverse temperature $\beta$ if:
\[
\text{for all local functions }f,\quad\bbP(f)=\lim_{\Lambda_n\uparrow\Z^d}\bbP_{q,\beta,\Lambda_n}^{\omega_n}(f)
\]
 for some sequence of finite
  volumes $(\Lambda_n)_n\uparrow\Z^d$ and of deterministic boundary conditions
$(\omega_n)_n\in\Omega$. We write $\bbP=\wlim_{n\to\infty} \bbP_{q,\beta,\Lambda_n}^{\omega_n}$.
Let $\calW_{q,\beta}$ be the space of weak-limiting Gibbs
states for the $q$-state Potts model.
\end{definition}

The non-emptiness of the set of DLR states follows from a compactness
argument in general \cite{Geo88}, but for the Potts model this can be
proved constructively.
For $i\in\{1,\ldots,q\}$, the weak limits $\wlim_{\Lambda\uparrow\Z^d}
\bbP^{i}_{q,\beta,\Lambda}$ exist and belong to $\calG_{q,\beta}$ (in particular, the limit does not depend
on the sequence of boxes chosen); this follows easily, e.g., from the random
cluster representation~\cite{Gri2006}. We denote by $\bbP^{i}_{q,\beta}$ the corresponding
limit. It can be checked~\cite[Prop.~6.9]{GeoHagMae01} that the phases
$\bbP^{i}_{q,\beta}$ are
translation invariant.

When $\beta$ is less than the critical inverse temperature
$\beta_c=\beta_c(q,d)$ (which is non-trivial for $d\geq2$), it is known that there exists a
unique infinite-volume Gibbs measure. The
relevant values of $\beta$ for a study of $\calG_{q,\beta}$ are thus $\beta\ge
\beta_c(q,d)$. 

\subsection{The case of the Ising model}

Let $\mu^+$ and $\mu^-$ be the two pure phases (that is,
translation-invariant extremal Gibbs measures) of the Ising model. Let
$\mu^\pm$ denote the limiting infinite-volume Gibbs state for the
Dobrushin boundary condition $\omega^\pm$ such that
  $\omega^\pm_{(x,y,z)}=+1$ for $z\geq0$ and $-1$ for $z<0$. We write
$\mu^\mp$ for the ``spin flip'' of $\mu^\pm$, namely the measure symmetric with respect to the plane $z=-1/2$.

\subsubsection{Dimension 2} In the beginning of the 80s, Aizenman
\cite{Aiz80} and Higuchi \cite{Hig81} proved independently that the
DLR states of the 2d Ising model are all convex combination of the
pure phases, namely, for any $\beta\geq0$,
\begin{equation}\label{AH}
\calG_{2,\beta}=\{\alpha\mu^++(1-\alpha)\mu^- : \alpha\in[0,1]\}.
\end{equation}
In particular, all the DLR states are translation invariant:
$\rm{tr}\calG_{2,\beta}=\calG_{2,\beta}$.

Gallavotti~\cite{Gal72}
proved, by studying the fluctuations of the Dobrushin interface, that
the corresponding weak limiting state $\mu^\pm$ is the mixture $\tfrac12(\mu^+ +
\mu^-)$. This was refined by Higuchi~\cite{Hig79}, who proved that the interface, after
diffusive scaling, weakly converges to a Brownian bridge at sufficiently low
temperatures. These two results were then pushed to all subcritical temperatures
by, respectively, Messager and Miracle-Sole~\cite{MesMir1977} and Greenberg and
Ioffe~\cite{GreIof05}.

By exploiting the Gaussian scaling of the Dobrushin interface, Abraham
and Reed \cite{AbrRee1976} produced a set of deterministic boundary
conditions $(\omega_\alpha)_{\alpha\in(0,1)}$ such that
$\wlim_{n\to\infty}\mu_{\Lambda_n}^{\omega_\alpha}=\alpha\mu^++(1-\alpha)\mu^-$.
Basically, they shift up the Dobrushin boundary condition by an amount
$C_\alpha\sqrt{n}$ around the cubic box of size $n$, and choose the
right constant $C_\alpha$ to get the mixture with proportion $\alpha$
of $\mu^+$. These results imply that the weak limiting states and the
DLR states of the Ising model coincide in 2 dimensions: for any $\beta\geq0$,
\begin{equation*}
\calW_{2,\beta}=\calG_{2,\beta}.
\end{equation*}

Note that the behavior of the macroscopic interfaces induced by an arbitrary
boundary condition was studied in \cite{CoqVel2010}. We refer to
\cite{BodIofVel2000} for a review on the microscopic theory of
equilibrium crystal shapes.

\subsubsection{Dimension 3 (and more)}
The existence of non-translation invariant states in dimension 3 and
more was discovered by
Dobrushin \cite{Dob72}. He proved that, at low enough
temperatures, the interface created under $\mu^\pm_{\Lambda_n}$ is
rigid, namely given by a plane with
local defects, and the
corresponding weak limiting Gibbs state is extremal. 
This implies in particular the existence of a countable number of
extremal DLR states in dimension $d\geq3$ at low enough temperature,
which are in
bijection with all the hyperplanes of $\Z^d$ orthogonal to any
coordinate axis. It is however widely believed that the 3-dimensional system has a ``roughening-
temperature'' $1/\beta_R$ above which the horizontal interface is no longer
sharp, and the corresponding Gibbs state is translation invariant.

The horizontal Dobrushin states are conjectured to be the only extremal
non-translation invariant states in 3 dimensions. We quote \cite{MesMir1977}: ``there can only be planes parallel to the faces of the lattice
cubes at finite distance, and no angles, corners, or diagonal planes
as rigid interfaces.''
For example, the 3d Dobrushin interface orthogonal to the vector $(1,1,1)$ is believed to be
delocalized, and to have $O(\sqrt{\log n})$ fluctuations in finite
volume at low temperature, where $n$ is the side length of the box. 
The result is currently known only\footnote{We emphasize that the natural
monotonicity of the fluctuations that we could expect with respect to
the temperature is not true in general. Indeed, positive temperature may result in reduction of the
fluctuations~\cite{BodGiaVel2001}.} at zero temperature
  \cite[Theorem 15]{Ken1997}.
 Note that the similar
diagonal Dobrushin interface in 4 and more dimensions (orthogonal to the vector $(1,1,1,\ldots,1)$) is rigid at low
enough temperature \cite{MesMir1977}, which enriches the set
$\rm{ex}\calG$. 

It is an interesting question to determine what the typical
fluctuations of the interfaces enforced by general boundary
conditions are, in particular those giving rise to non-planar limiting
shapes. This is in general already an
open problem at zero temperature, and for an isotropic surface
tension.
The best known results in this direction are large deviation principles. Cerf and Pisztora \cite{CerPis2001} proved that, in
dimensions $d\geq3$, for a
given ``macroscopic'' boundary condition\footnote{The boundary of a
  fixed region $\Omega\subset\R^d$ must be partitioned in such a way that for each $n$, on the boundary of
  $\Omega_n=\Omega\cap\frac1n\Z^d$, the number of
  nearest-neighbor pairs of vertices having different spins
  is $o(n^{d-1})$. },
asymptotically as the mesh size of the box tends to zero, the law of the so-called phase
partition (i.e.\ the partition of the space according to the value of
the locally\footnote{More precisely in a region of size $f(n)$ such
  that $\log n\ll f(n)\ll n^{1/(d-1)}$.} dominant spin) is determined by a variational problem. More
precisely, the empirical phase partition is $\ep n$-close to some partition
which is compatible with the boundary condition and minimizes the
surface tension. It is conjectured that, as $\beta\downarrow\beta_c$,
the (rescaled) surface tension becomes more and more isotropic and so the
solution of the variational problem should approach the solution of
the classical (isotropic) Plateau problem.

Concerning $\rm{tr}\calG$, Bodineau \cite{Bod2006} proved that for any $d\geq3$ all the
translation invariant Gibbs states of the Ising model are convex
combinations of the pure phases $\mu^+$ and $\mu^-$.

Let us now summarize which consequences these known results have on
the sets $\calW$ and $\calG$ in dimension 3; see Figure
\ref{inclusion}.

\begin{figure}[h]
  \centering
 \includegraphics[width=10cm]{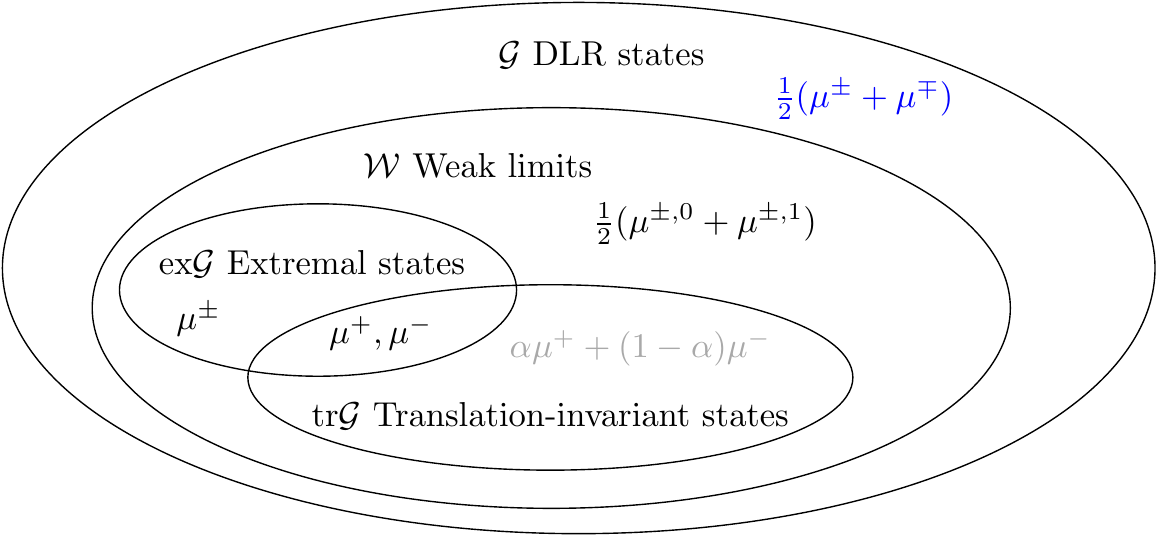}
 \caption
   {Inclusion of properties, and examples for the Ising model on $\Z^3$ at
     low temperature. It is a still a conjecture that
     all mixtures of $\mu^+$ and $\mu^-$ are weak limits. We prove the outmost
     result.}
\label{inclusion}
\end{figure}

 If the conjecture about the fluctuations of the low-temperature
 tilted Dobrushin interface is true, then the
 corresponding Gibbs state in the thermodynamic limit is translation
 invariant, and an
argument \`a la Abraham and Reed \cite{AbrRee1976} allows to
construct a sequence of boundary conditions which have
$\alpha\mu^++(1-\alpha)\mu^-$ as weak limit, for any $\alpha\in(0,1)$.
One has to shift
up the plane by an amount $C_\alpha\sqrt{\log n}$. Together with Bodineau's
characterization \cite{Bod2006} of the translation invariant states, this would imply
that $\rm{tr}\calG\subset\calW$.

Note that
there exist mixtures of non-translation invariant states which
are reachable with boundary conditions. Let us denote by $\mu^{\pm,z}$
the Ising measure with horizontal Dobrushin boundary condition, parallel to
the plane $xy$ and at height $z$, then
$\mu=\frac12(\mu^{\pm,0}+\mu^{\pm,1})$ is the weak limit of
the ``one-step boundary condition'':
$$\omega_{(x,y,z)}=\left\{\begin{array}{l}+1 \text{ if } z\geq0 \text{ and if } z\geq-1
\text{ and }x\geq0,\\
-1 \text{ otherwise}.\end{array}\right.$$

Indeed, at low enough temperature, the horizontal Dobrushin interfaces are
localized, and so the typical interface induced by the ``one-step
boundary condition'' consist of a plane at height $-1$ inside the
half-space $x\geq0$, a plane at height $0$ inside the
half-space $x<0$, both with local defects, and a
one-dimensional step between the two which undergoes Brownian
fluctuations; see Figure \ref{step}. The associated Gibbs state is invariant under the
translations parallel to the $xy$ plane. General ``step
boundary conditions'' and their link with facets of the equilibrium crystal are studied in \cite{Mir1995}, see in particular
Remark 7.

\begin{figure}[h]
  \centering
 \includegraphics[width=6.5cm]{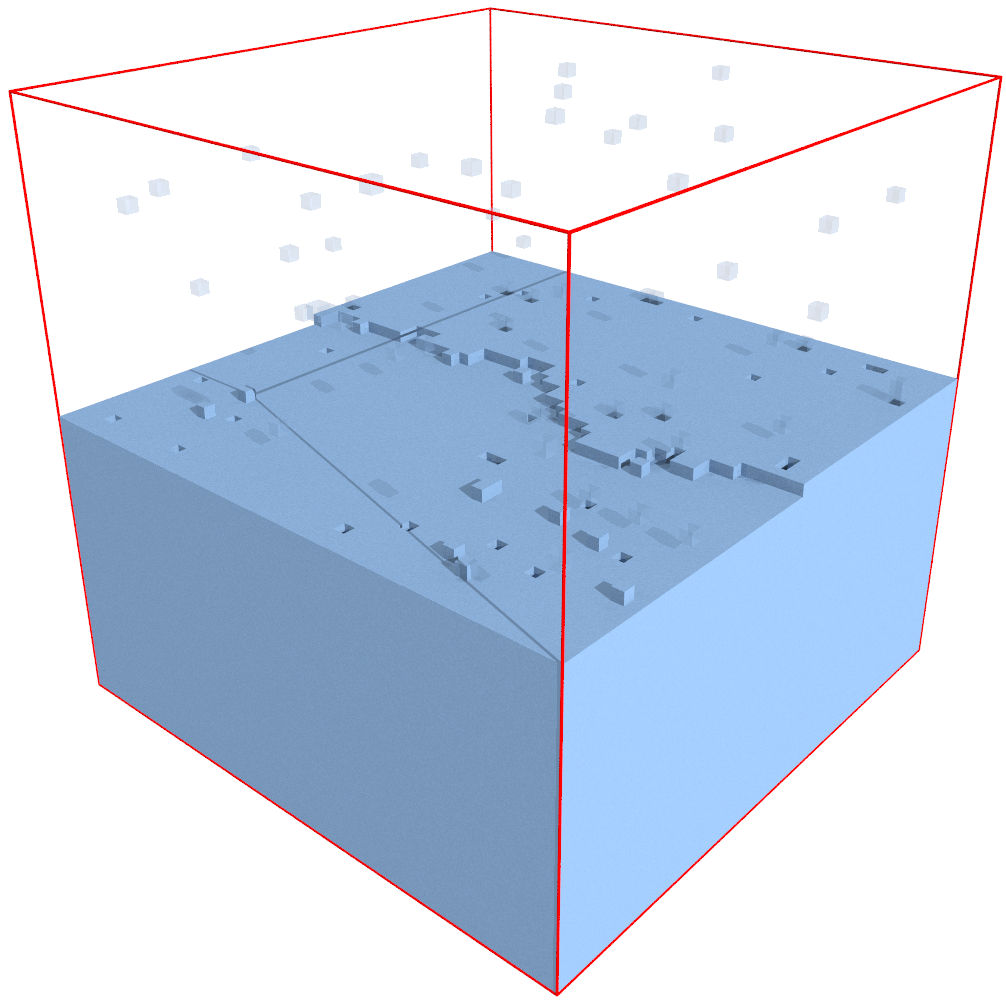}
 \caption
   {A realization of the ``one-step boundary condition'' at low
     temperature. Minus spins are blue (and translucent) cubes, and plus spins are
     transparent. 
     Simulation due to V.\ Beffara.}
\label{step}
\end{figure}

In this paper, we prove that there also exist mixtures of non-translation invariant states which
are \emph{not} reachable with boundary conditions. The proof is
presented in Section \ref{proof}.
\begin{theorem}\label{result-ising}
In dimension $d\geq3$, for $\beta$ large enough (depending on $d$),
\begin{equation*}\mu=\frac12(\mu^\pm+\mu^\mp)\in\calG_{2,\beta}\backslash\calW_{2,\beta}.
\end{equation*}
Namely $\mu$ cannot be reached by a sequence of finite
volume measures with boundary conditions.
\end{theorem}

\subsubsection{Random boundary conditions} The Ising model with
boundary conditions sampled from the symmetric i.i.d.\ field $\{-1,
1\}^{\Z^2}$ has been studied by van Enter \emph{et al.~}\cite{vEnMedNet2002,
  vEnNetSch2005}. A corollary of their results is that for a typical boundary condition,
the probability of the set of configurations containing an interface
tends to zero in the infinite-volume limit, which is a stronger result
than the absence of translationally non-invariant states
\eqref{AH}. When $\Z^2$ is replaced by $\Z^d$, for $d\geq4$, it is expected that $\{\mu^-,\mu^+\}$ is the almost
sure set of limit measures (along the regular sequence of cubes). For
$d=2,3$, they conjecture that this set is
$\rm{tr}\calG=\{\alpha\mu^++(1-\alpha)\mu^- : \alpha\in[0,1])\}$.

An interesting result concerning the biased setting can be
found in \cite{Hig1978}. Higuchi proved that $\mu^-$ is the only limiting Gibbs state
corresponding to a sequence of boundary conditions $\omega_n$ such
that the density $n_+$ of $+$ spins is smaller than $3/8$ on $\partial\Lambda_n$
for every $\omega_n$. The fraction $3/8$ is optimal in the sense that
for any $\theta>3/8$, there exists a sequence of boundary conditions
such that $3/8<n_+\leq\theta$ and for which the limiting Gibbs
state is $\mu^+$.

\subsubsection{Global Markov Property}
It is worth noting that a mixture of Dobrushin measures similar to
$(\mu^\pm+\mu^\mp)/2$ provides an example of a DLR state failing to satisfy the
global Markov property. We refer to \cite{AlbZeg1992} for a review of
the role of this property in statistical mechanics. However, there are
extremal Gibbs measures constructed by Israel \cite{Isr1986} which
also lack the global Markov property, and thus the two properties
(lacking the global Markov property and not being a weak limit state) are
not the same.
A state $\bbP$ is said to satisfy the
  global Markov property if
\begin{equation}\label{GMP}
  \bbP(\cdot|\calF_{\Lambda^c})(\omega)=\bbP(\cdot|\calF_{\partial\Lambda})(\omega)\text{
  for any (not necessarily finite) set $\Lambda$. }
\end{equation}
  For spin
  systems with
  nearest-neighbor interaction it is a generalized version of
  \eqref{def-DLR}. Let $\tilde\mu^\mp$ be the Gibbs
  states obtained from $\mu^\pm$ by the reflection
  $(x,y,z)\to(x,y,-z)$. Note that $\tilde\mu^\mp\neq\mu^\mp$, since this
  reflection is the identity on the plane $z=0$, so that
  $\tilde\mu^\mp$ agrees with $\mu^\pm$ on $\calF_{\{z=0\}}$. Then the
  article \cite{Gol1980} explains that
  $\tilde\mu:=\frac12(\mu^\pm+\tilde\mu^\mp)$ does not satisfy
  \eqref{GMP} for $\Lambda=\{z>0\}$, since $\tilde\mu^\mp$ and
  $\mu^\pm$ agree on $\partial\Lambda$ but are mutually singular on
  $\Lambda^c$. For the proof of Theorem \ref{result-ising}, we also
  use the idea that specifying $\sigma$ on one side of the box determines
  whether $\sigma$ is a configuration of the first or the second phase
  of the mixture, but we need more input.

\subsection{The case of the Potts model}

\subsubsection{Dimension 2} The set $\calG_{q,\beta}$ for $\beta>\beta_c$ has been recently proved
to be the simplex with the $q$ pure phases as
extremal measures \cite{CoqDumIofVel2014}. In particular all the Gibbs states are translation
invariant. 
\begin{equation}\label{Gq}
\calG_{q,\beta}=\left\{\sum_{i=1}^q \alpha_i\bbP^i_{q,\beta} : \alpha_1,\ldots,\alpha_q\geq0,
\sum_{i=1}^q\alpha_i=1\right\}
\end{equation}
We summarize here the main results of the above work.

Although an arbitrary boundary condition $\omega_n$ can a priori enforce the presence of
$O(n)$ interfaces, we proved that, uniformly in $\omega_n$, only a
finite number of them penetrate up to the half box with high probability.
Moreover, these
macroscopic interfaces are in a $\delta n$ neighborhood of the graphs
which are solutions of the so-called Steiner problem: link the
endpoints in a way which is compatible with the boundary condition and
which minimizes surface tension.

These minimal graphs are called Steiner forests (they are collections
of disjoint trees).  Due to the uniform convexity
of the surface tension, proved in \cite{CamIofVel08} for all $q\geq2$,
and a
general geometric argument exposed in
\cite{AlfCon91}, each inner node of the trees has degree 3, and there exists an $\eta > 0$ such that the angle between two edges incident to an inner node is always larger than $\pi/2+\eta$.

As a consequence, the possible local configurations of the system in the
$\ep n$ neighborhood of the origin are either a pure phase, or two phases
separated by a straight interface (which undergoes Brownian
fluctuations \cite{CamIofVel08}), or three phases separated by a ``tripod-like''
interface (whose triple point and legs undergo Brownian
fluctuations).
The archetypical
illustration is the 1-2-3-4 boundary condition which gives rise to two
possible Steiner trees; see Figure \ref{simu1234}.

\begin{figure}[h]
  \centering
\includegraphics[width=5cm]{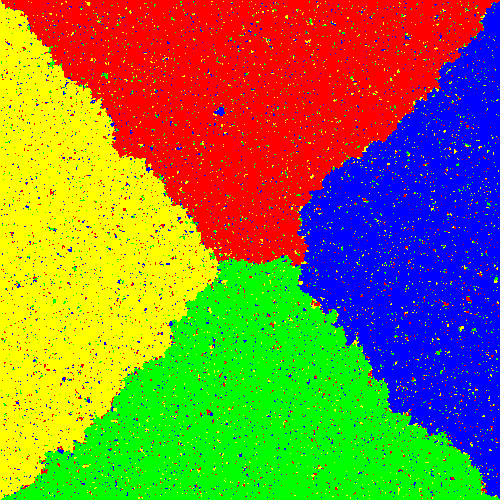}
\includegraphics[width=5cm]{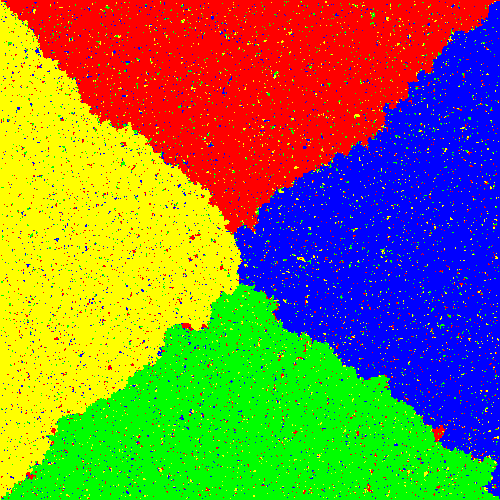}
\includegraphics[width=5cm]{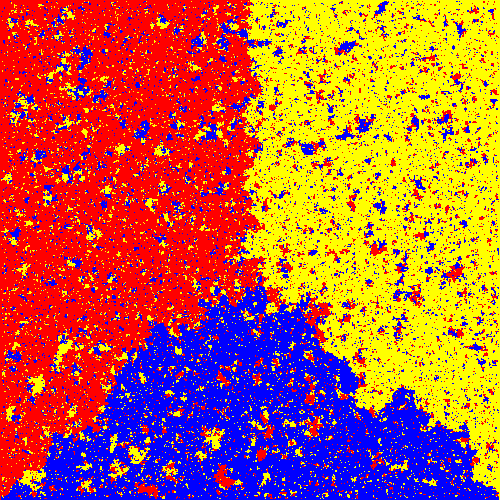}
 \caption
   {(Left, middle) Two realizations of the 1-2-3-4 boundary condition for the $q=4$
     Potts model. (Right) A realization of the 1-2-3 boundary condition
     for $q=3$ below the critical temperature. Simulation due to V.\ Beffara.}
\label{simu1234}
\end{figure}

Given these results, an argument \`a la Abraham and Reed should achieve to reach
mixtures of \emph{three} pure phases with boundary conditions : take the 1-2-3
boundary condition, and shift it with respect to the origin by a
vector $\mathbf{C}_{\mathbf{\alpha}}\sqrt{n}$, with a well-chosen
$\mathbf{C}=\mathbf{C}({\alpha_1,\alpha_2,\alpha_3})$ in order to bias
the limiting measure towards
$\alpha_1\bbP^1+\alpha_2\bbP^2+\alpha_3\bbP^3$; see Figure
\ref{simu1234} on the right. 

This would imply
that for any $\beta\geq0$,
\begin{equation*} 
\calW_{q,\beta}=\calG_{q,\beta}, \text{ for }q=2,3.
\end{equation*}

Starting at $q=4$, asking what the structure of $\calW$ is
(and if $\calW=\calG$) becomes a difficult
question. The study of Steiner forests gives one way to construct non-trivial
convex combinations of pure phases: we can look for symmetric domains and boundary
conditions which give rise to several possible Steiner trees
intersecting at some location. However, on the one hand quite little is known about
the structure of Steiner forests for a general norm on the plane, and
on the other hand, it is not clear if we can get all the
weak limiting states with this method. 

However, by adding a slowly growing number of boundary spins to the free
boundary condition, it might be possible to obtain continuous changes
of weights (thus biasing the mixture $\frac1q\sum_{i=1}^q\bbP^i$). Therefore, it seems to be reasonable to conjecture that
$\calW_{q,\beta}=\calG_{q,\beta}$ for the 2-dimensional Potts model
for any $q$. This question will be addressed in a forthcoming paper.

\subsubsection{Dimension 3 (and more)}
The large deviations results of Cerf and Pisztora \cite{CerPis2001} for the empirical
phase partition, which we already
mentioned in the previous subsection,  are valid for the Potts models for all $q\geq2$ in
dimension 3 below the critical temperature. For $q\leq4$, it is
conjectured that the (rescaled) surface tension converges to the Euclidean ball
as $\beta\downarrow\beta_c(q)$,
whereas for $q$ large (conjecturally up to $q=4$), this should not be the case as the phase transition
is of first order.

The macroscopic
phase separation surfaces are minimizing the surface tension $\tau$. Note that the geometry of interfaces is much more complicated in
systems with more than two phases. 
Moreover, very little is known about the surface tension in dimension 3, although
the following properties are
widely believed to be true: $\tau$ satisfies the sharp simplex
inequality (that is $\tau$ is uniformly convex), the value of $\tau$ is minimal in axis directions, and
$\tau$ increases as the normal vector moves from $(0,0,1)$ to
$(1,1,1)$. See the introduction of \cite{CerPis2001}.

The localization of the horizontal Dobrushin interface for the Potts model
($\omega_{(x,y,z)}=i$ if $z\geq0$ and $j$ if $z<0$) at low enough
temperature for $d\geq3$ has been proved by
Gielis and Grimmett \cite{GieGri2002}. 
A non translation invariant measure, corresponding to the coexistence of the ordered and disordered phases, is obtained by \v Cern\'y and
Koteck\'y \cite{CerKot2003} for the random cluster model at the
criticality $p_c(q)$ with sufficiently large $q$ (limit of finite
volume measures with free-wired Dobrushin boundary conditions).

Concerning $\rm{tr}\calG$, Martirosian \cite{Mar86} proved that for
any $d\geq2$ and $q$ large enough (depending on $d$) all the
translation invariant Gibbs states of the Potts model at $\beta>\beta_c(q,d)$ are convex
combinations of the pure phases $\bbP^i_{q,\beta}$, $i=1,\ldots,q$.

The state of the art concerning $\calW$ is thus even more restricted
than in the case of the Ising model. A natural guess is that a
similar result as Theorem \ref{result-ising}, namely
$\frac12(\bbP^{12}_{q,\beta}+\bbP^{21}_{q,\beta})\in\calG_{q,\beta}\backslash\calW_{q,\beta}$, is true for the Potts
model. Nevertheless, we argue in Section \ref{potts} why it is much
more difficult to prove (or disprove). Both outcomes would of course
be interesting.

\section{Proof of Theorem \ref{result-ising}}\label{proof}

Write as above $\mu=\frac12(\mu^\pm+\mu^\mp)$. 
We first use the
localization of the Dobrushin interface \cite{Dob72} in dimension
$d\geq3$ to deduce positive association of flipped spins
across the symmetry plane $z=0$. Using the bound
of Van Beijeren \cite{vBe1975} on the magnetization of a spin at
height 0, and the FKG inequality, we have
\begin{align}
m^\star_2&\leq\mu^\pm(\sigma_{(x,y,z)})\leq m^\star_3\hspace{1cm}\text{ for
}z\in\N^+,
\end{align}
where $m^\star_d=\mu^+(\sigma_0)$ in dimension $d$. Note that
$m^\star_d\to1$ as $\beta\to\infty$ for all $d$. Moreover, by symmetry,
$$ \mu^\pm(\sigma_{(x,y,z)}=-1)=\mu^\mp(\sigma_{(x,y,z)}=+1),$$
and hence
\begin{equation}\label{sym}
\mu(\sigma_{(x,y,z)}=+1)=\mu(\sigma_{(x,y,z)}=-1)=1/2.
\end{equation}
Let $\bfz=(0,0,z)$ for some $z\in\N^+$ and $\hat\bfz=(0,0,-z-1)$, two
points which are symmetric with respect to the plane $z=-1/2$.
Note that by symmetry $\mu^\pm(\sigma_\bfz)=-\mu^\pm(\sigma_{\hat\bfz})$.
By a union bound and \eqref{sym}, we have 

\begin{align}\label{loc}
\mu(\sigma_{\bfz}=+1\,|\,\sigma_{\hat\bfz}=-1)\nonumber
&=\frac{\mu(\sigma_{\bfz}=+1\,,\,\sigma_{\hat\bfz}=-1)}{\mu(\sigma_{\hat\bfz}=-1)}
\geq\frac{\frac12\mu^\pm(\sigma_{\bfz}=+1\,,\,\sigma_{\hat\bfz}=-1)}{1/2}\nonumber\\
&\geq 1-\mu^\pm(\sigma_{\bfz}=-1)-\mu^\pm(\sigma_{\hat\bfz}=+1)\nonumber\\
&=
1-\frac{1-\mu^\pm(\sigma_{\bfz})}2-\frac{1+\mu^\pm(\sigma_{\hat\bfz})}2
=\mu^\pm(\sigma_{\bfz})\geq m^\star_2
\end{align}

Now suppose that $\mu$ is a weak limit of finite-volume measures,
i.e. $\mu=\lim_{n\to\infty}\mu^{\omega_n}_{\Lambda_n}$ for some
deterministic sequence of boundary conditions $(\omega_n)_n$ and
$\Lambda_n\uparrow\Z^d$. As every $\mu^{\omega_n}_{\Lambda_n}$ satisfy
the FKG inequality, so does $\mu$. Which implies,
\begin{equation}
\mu(\sigma_{\bfz}=+1|\sigma_{\hat\bfz}=-1)\leq
\mu(\sigma_{\bfz}=+1)
=1/2.
\end{equation}
This is a contradiction with \eqref{loc} as soon as
$\beta$ is large enough that $m^\star_2(\beta)>1/2$. Note that if we take $z$ large, we can actually
replace the bound in \eqref{loc} by $m^\star_3(1-\ep)$, with some
$\ep=\ep(\beta,z)\to0$ as $z\to\infty$. The
contradiction holds then as soon as $\beta$ is large enough that $m^\star_3(\beta)>1/2$.
\qed

\begin{remark}\label{proof2}
 We could have used ``negative association of the same value of spin''
 across the Dobrushin interface, namely the
 following inequality holds as well:
\begin{align}
\mu(\sigma_{\bfz}=+1|\sigma_{\hat\bfz}=+1)
\leq1-m^\star_2
\end{align}
and is in contradiction with the FKG inequality:
\begin{equation*}
\mu(\sigma_{\bfz}=+1|\sigma_{\hat\bfz}=+1)\geq
\mu(\sigma_{\bfz}=+1)=1/2.
\end{equation*}
As we will see in the next section, this ``second proof'' gives a
priori two hopes of extending the result to the Potts model. However,
none of them works.
\end{remark}
\begin{remark}
{The result holds for all vertical translates and axis-symmetry of the Dobrushin
  boundary condition, as well as for
  $\mu=\alpha\mu^\pm+(1-\alpha)\mu^{\mp}$, with $\alpha\in(0,1)$ to be
  chosen such that}
$$\frac{m^\star_2}{\frac{1+m^\star_3}2+\frac{1-\alpha}\alpha\frac{1-m^\star_2}2}>\frac12.$$
\end{remark}

\begin{remark} 
It is possible to make the contradiction hold up to the
  roughening temperature $\beta_R$ of the 3-dimensional Ising model,
  by looking at the proportion of $+$ spins in large but finite boxes
  in the two half-spaces.
Recall that by
  \cite{vBe1975} we have $0<\beta_c(3)\leq\beta_R<\beta_c(2)$.

Let $\Lambda_m(\bfz)$ be the box of (odd) side-length $m< z$ centered at
$\bfz$. Denote by $M^m_\bfz$ the majority of the spins inside
$\Lambda_m(\bfz)$, namely
$$ M^m_\bfz=\left\{\begin{array}{ll}+1&\text{if }\sharp\{i\in\Lambda_m(\bfz) :
  \sigma_i=+1\}>\sharp\{i\in\Lambda_m(\bfz) : \sigma_i=-1\}\\
-1&\text{else,}\end{array}\right.$$
where $\sharp X$ denotes the cardinality of the set $X$. Then, by the
same computation as in \eqref{loc}, we have on the one hand
\begin{equation*}
\mu(M^m_\bfz=+1,M^m_{\hat\bfz}=-1)\geq
\frac12(1-2\mu^\pm(M^m_\bfz=-1))\geq\frac12(1-2\ep),
\end{equation*}
with $\ep=\ep(z,m)$ being small in $z$ and $m$ as soon as
$\beta>\beta_R$.
And on the other hand, always by symmetry, $\mu(M^m_{\hat\bfz}=-1)=\frac12$.
So that
\begin{equation}\label{locM}
\mu(M^m_\bfz=+1\,|\,M^m_{\hat\bfz}=-1)\geq1-2\ep.
\end{equation}
As the event $\{M^m_\bfz=+1\}$ (resp. $\{M^m_{\hat\bfz}=-1\}$) is
increasing (resp. decreasing), the FKG inequality implies
\begin{equation*}
\mu(M^m_\bfz=+1\,|\,M^m_{\hat\bfz}=-1)\leq \mu(M^m_\bfz=+1)\leq\frac12,
\end{equation*}
which is in contradiction with \eqref{locM} as soon as
$\beta>\beta_R$, if $z$ and $m$ are taken sufficiently large.
\qed
\end{remark}

\section{The state of affairs for the Potts model}\label{potts}

\subsection{Absence of certain spin correlation inequalities in
  presence of boundary conditions}

The direct generalization of the proof
of Theorem \ref{result-ising} to the Potts model is not possible,
because the needed correlation inequalities
break down for non-free boundary conditions. Indeed, we provide counter-examples to the FKG inequality for the fuzzy Potts measure with
non-free boundary conditions; see \cite{KahWei2007} for the case of
free boundary conditions.


 First, for the Potts model with free
  boundary conditions, on any finite graph $\Lambda$, Schonmann \cite{Sch1988} proved that the following correlation inequality
  holds, which we could name ``negative association of different kinds
  of spins''. For any $i\neq j\in\{1,\ldots,q\}$, and $A,B\subset\Lambda$,
\begin{equation}\label{corr-ab}
\bbP^{\varnothing}_{q,\beta,\Lambda}\left(\prod_{x\in A}\ind{\sigma_x=i} \middle| \prod_{y\in
  B}\ind{\sigma_y=j}\right)\leq
\bbP^{\varnothing}_{q,\beta,\Lambda}\left(\prod_{x\in A}\ind{\sigma_x=i} \right).
\end{equation}

\begin{proposition}
There exist boundary conditions $\omega$ such that \eqref{corr-ab} does not hold for $\bbP^\omega_{q,\beta,\Lambda}$.
\end{proposition}

\begin{proof}
Here is a counter-example, based on the
analysis of subcritical Gibbs states of the Potts model on $\Z^2$.
For the boundary condition $\omega=$1-2-3-4 depicted in Figure \ref{steiner}, at fixed supercritical
$\beta$, in a sufficiently large box, the typical interfaces are
concentrated around two possible deterministic Steiner trees, and undergo Brownian
fluctuations around these objects; see also Figure \ref{simu1234}. 

\begin{figure}[h!]
  \centering
 \includegraphics[width=4.7cm]{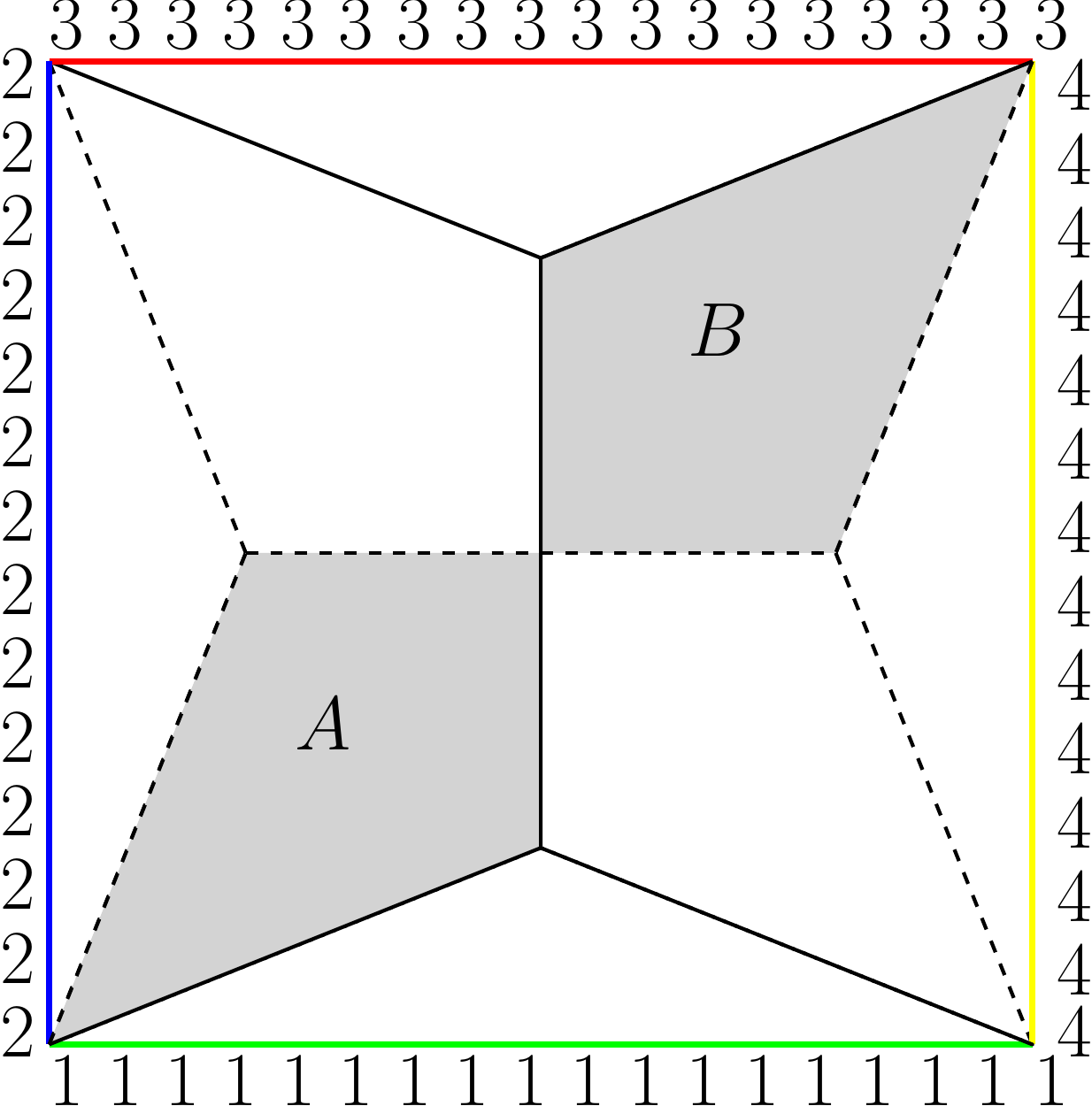}
 \caption
   {The two possible Steiner trees (solid and dashed lines) for the 1-2-3-4 boundary condition.}
\label{steiner}
\end{figure}

Indeed, by uniform convexity of the
  Wulff shape \cite{CamIofVel08}, these two trees are shorter
than the spanning minimal tree (consisting of three sides of the box),
which would have 90 degrees between its branches; see \cite{AlfCon91}. Therefore, for some $x$ in the
region $A$ and some $y$ in the region $B$, at large enough $n$, we have :
\begin{equation}
\bbP^{1234}_{\Lambda_n}(\sigma_x=1|\sigma_y=3)\geq1-\ep
\quad\text{ but }\quad
\bbP^{1234}_{\Lambda_n}(\sigma_x=1)\leq \frac12+\ep<1,
\end{equation}
with $\ep=\ep(\beta)\to0$ as $\beta\to\infty$, which contradicts
\eqref{corr-ab}, and also shows that the analogue of the
``first'' proof for the Ising model cannot be extended to the Potts
model.
Counter-examples of this kind exist for $q=3$.
\end{proof}

 Secondly, Schonmann \cite{Sch1988} also proves ``positive association of the same kind
  of spins'' for the Potts model
  with free boundary conditions on any finite graph. That is, for any $i\in\{1,\ldots,q\}$,
\begin{equation}\label{corr-aa}
\bbP^{\varnothing}_{q,\beta,\Lambda}\left(\prod_{x\in A}\ind{\sigma_x=i} \middle| \prod_{y\in
  B}\ind{\sigma_y=i}\right)\geq
\bbP^{\varnothing}_{q,\beta,\Lambda}\left(\prod_{x\in A}\ind{\sigma_x=i}\right ).
\end{equation}
However, if soap-film-like surfaces are minimal surfaces for the surface tension of the
3 dimensional Potts model (which is widely believed to be true), then
\eqref{corr-aa} does not hold for the 1-2-3-4-5-6 as boundary condition, namely a different color on each
face of a cube, see Figure \ref{steiner-soap}. The typical interfaces should be
concentrated around one of three possible minimal surfaces, each one
having a little square aligned with one coordinate axis. For some $x$ and $y$ in well-chosen regions (more precisely in the interior of two different parts of the symmetric difference between
the locations of a phase in two Steiner surfaces), we can get
$\bbP^{123456}_{\Lambda_n}(\sigma_x=1|\sigma_y=1)\leq 1/2+\ep$ whereas $\bbP^{123456}_{\Lambda_n}(\sigma_x=1)\geq 2/3-\ep>1/2$
which contradicts \eqref{corr-aa}, and shows that the analogue of the
``second'' proof for the Ising model, mentioned in Remark \ref{proof2}, cannot be extended to the Potts
model.

\begin{figure}[h!]
  \centering
\begin{tabular}{cc}
\includegraphics[width=4.6cm]{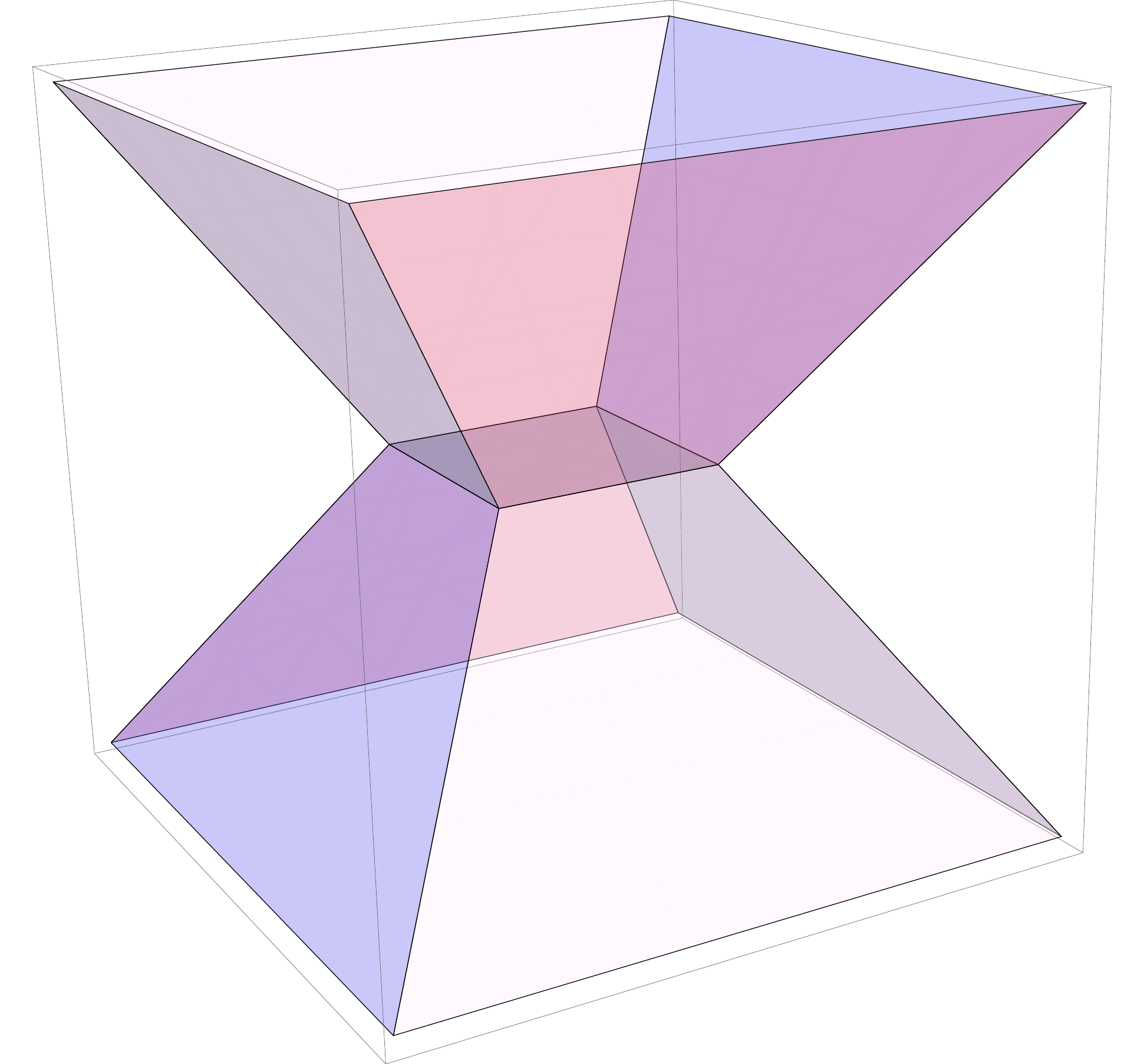}
\includegraphics[width=4.2cm]{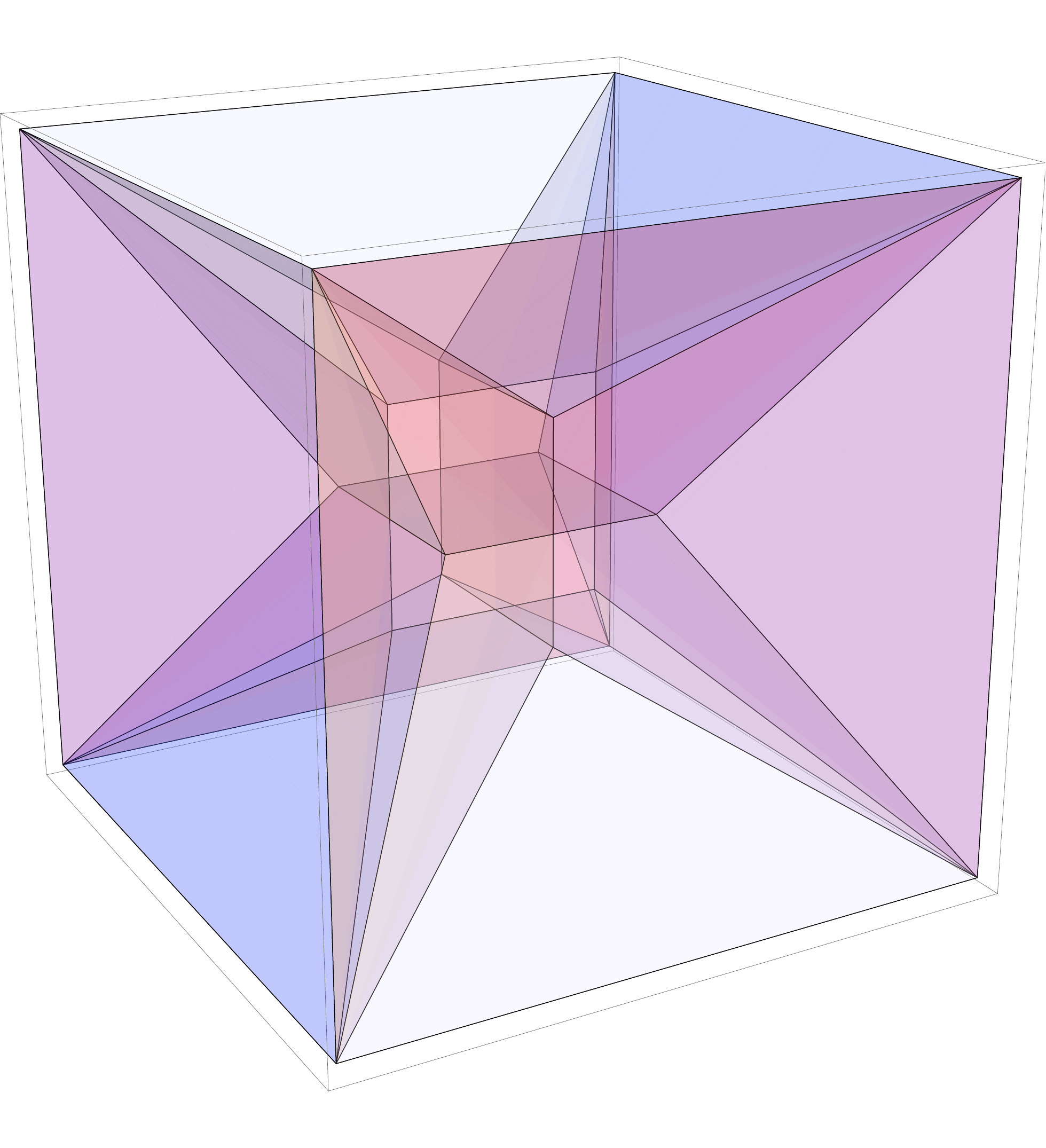}
\end{tabular}
 \caption
   {(Left) the soap film in direction $z$; (Right) the three possible
     soap films in directions $x,y$ and $z$.
}
\label{steiner-soap}
\end{figure}


Note that the existence of at least two Steiner
trees for which the locations of a pure phase have a non-empty
symmetric difference is enough to provide a 2d counter-example to
\eqref{corr-aa}.

\subsection{Correlation inequalities for specific boundary conditions
  and exclusion of certain weak limits}

 In \cite{VdBHagKah2006}, van den Berg \emph{et al.} proved some conditional correlation
  inequalities for the random cluster model on finite
  graphs $\Lambda=(V,E)$; 
see \cite{Gri2006} for the definition and a review on the
  random cluster model. For $S\subset V$, let $C_S$ denote the set of edges belonging to open paths starting
at vertices of $S$. They show that for $q\geq1$, if $S$ and
  $T$ are disjoint sets of vertices and $f$ and $g$ functions of the
  clusters of $S$ and $T$, written $(C_S,C_T)$, each increasing in
  $C_S$ and decreasing in $C_T$, then,
\begin{equation}\label{cond-corr-rcm}
\phi_{q,p,\Lambda}(fg | S\nleftrightarrow T)\geq \phi_{q,p,\Lambda}(f |
S\nleftrightarrow T) \cdot \phi_{q,p,\Lambda}(g | S\nleftrightarrow T)
\end{equation}
where $\phi_{q,p}$ denotes the random cluster model with parameters
$q$ and $p$ on the graph $\Lambda$.

We prove here that this result implies the
correlation inequality \eqref{corr-ab} in the Potts model with certain
specific boundary conditions. 

\begin{proposition}\label{corr-bicolor}
For $q\geq 2$, subsets $A,B\subset\Lambda\Subset\Z^d$, for any $i\neq j\in\{1,\ldots,q\}$ and any
bicolor boundary condition $\omega\in\{i,j,\varnothing\}^{\partial\Lambda}$, we have
\begin{equation}\label{corr-ab-omega}
\bbP^{\omega}_{q,\beta,\Lambda}\left(\prod_{x\in A}\ind{\sigma_x=i} \middle| \prod_{y\in
  B}\ind{\sigma_x=j}\right)\leq
\bbP^{\omega}_{q,\beta,\Lambda}\left(\prod_{x\in A}\ind{\sigma_x=i} \right).
\end{equation}
\end{proposition}

\begin{proof}
The well-known Edwards-Sokal
coupling \cite{EdwSok1988} implies
that the Potts measure $\bbP^\omega_{q,\beta,\Lambda}$ is coupled to the
random cluster measure
$\phi_{q,p,\Lambda}(\cdot | \cap_{i\neq j} \{E_i(\omega)\nleftrightarrow E_j(\omega)\})$
with $p=1-e^{-\beta}$ and $E_i(\omega)=\{x\in\partial\Lambda : \omega_x=i\}$.
Let us write $cond(\omega)=\cap_{i\neq j}
\{E_i(\omega)\nleftrightarrow E_j(\omega)\}$ and $\kappa(X)$ for the number of connected components of the set
$X$. Then,
\begin{align*}
\bbP^{\omega}_{q,\beta,\Lambda}\left(\prod_{x\in A}\ind{\sigma_x=i}  \prod_{y\in
  B}\ind{\sigma_x=j}\right)&=
\phi_{q,p,\Lambda}\left(f_A\cdot f_B | cond(\omega)\right)\\
\text{with }\quad f_A&=\sum_{X\subset A}\sum_{c=0}^{|A\backslash
    X|}\frac1{q^c}\ind{X\leftrightarrow E_i}\ind{A\backslash
    X\nleftrightarrow E_i}\ind{\kappa(A\backslash X)=c}\\
f_B&=\sum_{Y\subset B}\sum_{c'=0}^{|B\backslash
    Y|}\frac1{q^{c'}}\ind{Y\leftrightarrow E_j}\ind{B\backslash
    Y\nleftrightarrow E_j}\ind{\kappa(B\backslash Y)=c'}
\end{align*}
The key remark is that, for any $q\geq 2$, for any
bicolor boundary condition $\omega\in\{i,j,\varnothing\}^{\partial\Lambda}$, we have 
$cond(\omega)=\{E_i\nleftrightarrow E_j\}$, and so
\eqref{cond-corr-rcm} ensures that $f_A$ and $f_B$ are negatively
correlated. Indeed, one can check that $f_A$ is ``increasing in the
connectedness'' of the graph $A\cup E_i$ (resp.\ $f_B$ is ``increasing in the
connectedness'' of the graph $B\cup E_j$), which implies that
 $f_A$ is increasing in $C_{E_i}$ (and decreasing in $C_{E_j}$), and that
 $f_B$ is increasing in $C_{E_j}$ (and decreasing in $C_{E_i}$).
Therefore, using \eqref{cond-corr-rcm}, we get \eqref{corr-ab-omega}.
\end{proof}

Proposition \ref{corr-bicolor} gives a partial result for the Potts model
concerning the existence of non-weak limit states,
counterpart of Theorem \ref{result-ising}.\\

\begin{proposition} The measure $\bbP=\frac12(\bbP^{12}+\bbP^{21})$,
  mixture of Dobrushin states for the 3-dimensional Potts model, is not a weak limit of finite-volume measures with
boundary conditions $\omega\in\{1,2,\varnothing\}^{\partial\Lambda}$.
\end{proposition}
\begin{proof}
We adapt the ``first'' proof for the Ising model, and keep the same notations. Localization of the Dobrushin interface at low enough temperature is also known for the Potts
model \cite{GieGri2002}. Therefore,
\begin{align}\label{loc-potts}
\bbP(\sigma_{\bfz}=1|\sigma_{\hat\bfz}=2)
\geq 1-\ep
\end{align}
On the other hand, suppose that $\bbP$ is a weak limit of finite-volume
measures $\bbP^{\omega_n}_{q,\beta,\Lambda_n}$ for some
deterministic sequence of boundary conditions
$(\omega_n)_n\in\{1,2,\varnothing\}^{\partial\Lambda}$ and some boxes
$\Lambda_n\uparrow\Z^d$. By \eqref{corr-ab-omega}, every
$\bbP^{\omega}_{q,\beta,\Lambda}$ with $\omega\in\{1,2,\varnothing\}^{\partial\Lambda}$ satisfies
\begin{equation}\label{ferro1}
\bbP^{\omega}_{q,\beta,\Lambda}(\sigma_\bfz=1 | \sigma_{\hat\bfz}=2)\leq
\bbP^{\omega}_{q,\beta,\Lambda}(\sigma_\bfz=1 ).
\end{equation}
This inequality being preserved by weak limits, the measure $\bbP$
satisfies it as well, hence
\begin{equation}
\bbP(\sigma_{\bfz}=1|\sigma_{\hat\bfz}=2)\leq
\bbP(\sigma_{\bfz}=1)\approx\frac12\left(\bbP^1_{q,\beta}(\sigma_\bfz=1)+\bbP^2_{q,\beta}(\sigma_\bfz=1)\right)
\end{equation}
which converges to 1/2 as $\beta\to\infty$, providing a contradiction with \eqref{loc-potts}.
\end{proof}

Despite this result, it is important to mention that there is enough structure in the Potts model to possibly
allow mixtures of localized states. The measure $\bbP=\frac12(\bbP^{12}+\bbP^{21})$ might still be reachable by a
sequence of finite-volume measures with well-chosen boundary
conditions, for example having a non-trivial structure in 2
directions, allowing the intersection of different possible Steiner trees, and being translation invariant in the 3rd direction, in
order to localize them. This is a work in progress.

\noindent \paragraph{Acknowledgements.} 
I am grateful to A.~van Enter for
encouraging me to work on this question, and for valuable comments and
suggestions. I also thank Y.~Velenik for his advice, a stimulating discussion and a few references, A.~Bovier for
mentioning reference \cite{AlbZeg1992} to me, and V.~Beffara
for the simulations of Figures \ref{step} and \ref{simu1234}. 
I am indebted to Y.~Higuchi for pointing out reference \cite{Miy2004}.
This research was supported by the German Research Foundation (DFG) and the
Hausdorff Center for Mathematics (HCM).

\bibliographystyle{siam}
\small
\bibliography{bibliography}

\normalsize
\vfill

{\sc L.\ Coquille,
Institut f\"ur Angewandte Mathematik, 
Rheinische Friedrich-Wilhelms-Universit\"at,
Endenicher Allee 60, 53115 Bonn, Germany}

{\em E-mail address: \verb|loren.coquille@iam.uni-bonn.de|}

\end{document}